\documentclass{ifacconf}

\usepackage{graphicx}      % include this line if your document contains figures
\graphicspath{{visuals/pictures}} 

\usepackage{xcolor}
\usepackage{natbib}        % required for bibliography

\usepackage{amsmath}
\usepackage{amssymb}
\usepackage{mathtools}

\usepackage{tikz}
\usetikzlibrary{shadings}
\usetikzlibrary{calc}
\usetikzlibrary{intersections}
\usetikzlibrary{shapes,arrows}
\usetikzlibrary{backgrounds,fit,positioning}

\usepackage{pgfplots}
\usepgfplotslibrary{fillbetween}

\usetikzlibrary{positioning,plotmarks, matrix, arrows, calc, shapes}
\tikzstyle{blockdiag}	= [node distance=7mm, >=stealth', semithick]
\tikzstyle{block}		= [draw, rectangle, minimum width=1.05cm, minimum height=.8cm, align=center]
\tikzstyle{sum} 		= [draw,circle,inner sep=0pt, minimum size=5pt]
\tikzstyle{connector} 	= [draw,circle,inner sep=0pt, minimum size=0.01pt, fill=black, fill opacity=0,draw opacity=0]
\tikzstyle{connector2}  = [draw,circle,inner sep=0pt, minimum size=2pt, fill=black]
\tikzstyle{gain} 		= [draw, regular polygon, regular polygon sides=3, thick, minimum height=3em, minimum width=4em, rotate=30]
\tikzstyle{bguide} 		= [rectangle, minimum height=3em, minimum width=4em]
\tikzstyle{line} 		= [thick]
\tikzstyle{branch}	    = [circle,inner sep=0pt,minimum size=1mm,fill=black,draw=black]
\tikzstyle{guide} 		= [anchor=center]
\tikzstyle{box} = [draw=white, rectangle,minimum width=1.05cm, minimum height=.8cm,, fill=gray, fill opacity=0.2]
\tikzstyle{textbox} = [draw=white, draw opacity=0, rectangle,minimum width=1.05cm, minimum height=.8cm]

\newcommand{\norm}[1]{\left\|#1\right\|}
\newcommand{\abs}[1]{\left|#1\right|}
\newcommand{\field}[1]{\mathbb{#1}}
\newcommand{\R}{\field{R}}
\newcommand{\RH}{\field{RH}_\infty}

\newcommand{\bmtx}{\begin{bmatrix}}
	\newcommand{\emtx}{\end{bmatrix}}
\newcommand{\bsmtx}{\left[ \begin{smallmatrix}} 
	\newcommand{\esmtx}{\end{smallmatrix} \right]} 
\newcommand{\bmatarray}[1]{\left[\begin{array}{#1}}
	\newcommand{\ematarray}{\end{array}\right]}

\begin{document}
	\begin{frontmatter}
		
		\title{Finite-Horizon Robustness Analysis under Mixed Disturbances using Signal-IQCs} 
		% Title, preferably not more than 10 words.
		
		\thanks[footnoteinfo]{This research was funded by the German Federal Ministry for Economic Affairs and Climate Action under grant number 20Y2109E and partially supported by the European Union under Grant No. 101153910. The responsibility for the content of this paper is with its authors.}
		
		\author[First]{Frederik Thiele} 
		\author[First]{Harald Pfifer}
		\author[First]{Felix Biertümpfel}

		\address[First]{Chair of Flight Mechanics and Control, Technische Universität Dresden, 01069 Dresden, Germany (e-mail: \{frederik.thiele, harald.pfifer, felix.biertuempfel\}@tu-dresden.de) }
		%\address[Second]{Colorado State University, 
			%   Fort Collins, CO 80523 USA (e-mail: author@lamar. colostate.edu)}
		%\address[Third]{Electrical Engineering Department, 
			%   Seoul National University, Seoul, Korea, (e-mail: author@snu.ac.kr)}
		
		\begin{abstract}                % Abstract of 50--100 words
			Common worst-case analyses for uncertain finite-horizon systems consider quadratic performance metrics based on the strict Bounded Real Lemma. Thus, they assess system performance for bounded inputs, e.g., signals in L2, which exhibit a worst-case shape. Consequently, known disturbance characteristics are left unexploited and uncovered, leading to unnecessarily conservative results. The present paper develops a worst-case analysis covering arbitrarily L2-bounded worst-case signals and partially known disturbances simultaneously. This is achieved by modeling the latter using signal integral-quadratic constraints (IQCs). The resulting analysis condition relies on a dissipation inequality within the IQC framework for finite time horizon problems. 
			This framework also readily allows to incorporate additional system uncertainties in the analysis. The approach's feasibility is demonstrated with the worst-case performance analysis of a small unmanned aerial vehicle in an urban environment.  
			
			%The paper presents a novel approach for the robustness analysis of uncertain systems over finite time horizon under known and unknown disturbances. Common worst-case analyses derived from the strict Bounded Real Lemma are limited to the analysis of arbitrary L2 bounded disturbances. This yields unnecessary conservative results, if some disturbances are (partially) known a priori. The presented approach circumvents this shortcoming by employing signal-based integral-quadratic constraints (IQCs) to model known disturbances, while maintaining the worst-case nature of others. The calculation of the worst-case upper bound relies on the dissipation inequalities inside the IQC framework. A numerical example is provided to demonstrate the feasibility of the proposed approach. 
			
			%Exploiting a recent extension of the strict BRL to IQCs, the analysis can also cover system uncertainties.
			%The approach further includes a bound for the error associated with the time-varying linearization.  Hence, the results obtained in the linear framework provide a valid upper bound for the worst-case performance of the nonlinear system. The calculation of the upper bound relies on the dissipation inequalities formulated in the framework of integral quadratic constraints. It is therefore computationally much cheaper than sample-based methods such as Monte Carlo simulation. The feasibility of the approach is demonstrated on a numerical example.
		\end{abstract}
		
		
	\end{frontmatter}
	
	%% There are a number of predefined theorem-like environments in
	%% ifacconf.cls:
	%%
	%% \begin{thm} ... \end{thm}            % Theorem
	%% \begin{lem} ... \end{lem}            % Lemma
	%% \begin{claim} ... \end{claim}        % Claim
	%% \begin{conj} ... \end{conj}          % Conjecture
	%% \begin{cor} ... \end{cor}            % Corollary
	%% \begin{fact} ... \end{fact}          % Fact
	%% \begin{hypo} ... \end{hypo}          % Hypothesis
	%% \begin{prop} ... \end{prop}          % Proposition
	%% \begin{crit} ... \end{crit}          % Criterion
	
	%% Example Figure 
	%\begin{figure}
	%	\begin{center}
		%		\includegraphics[width=8.4cm]{bifurcation}    % The printed column width is 8.4 cm.
		%		\caption{Bifurcation: Plot of local maxima of $x$ with damping $a$ decreasing} 
		%		\label{fig:bifurcation}
		%	\end{center}
	%\end{figure}
	
	%% Example Table 
	%\begin{table}[hb]
	%	\begin{center}
		%		\caption{Margin settings}\label{tb:margins}
		%		\begin{tabular}{cccc}
			%			Page & Top & Bottom & Left/Right \\\hline
			%			First & 3.5 & 2.5 & 1.5 \\
			%			Rest & 2.5 & 2.5 & 1.5 \\ \hline
			%		\end{tabular}
		%	\end{center}
	%\end{table}
	
	%===============================================================================
	
	%%%%%%%%%%%%%%%%%%%%
	%%% ------------- %%
	%%% -- Section -- %%
	%%% ------------- %%
	%%%%%%%%%%%%%%%%%%%%
	
	\section{Introduction} 
	
	The induced $L_2$-norm is a well-established metric for evaluating the performance of key control objectives such as tracking, under the assumption of arbitrary norm-bounded inputs in $L_2$ over an infinite time horizon. In practice, however, engineering problems often involve external disturbances that are not in $L_2$, but instead follow specific temporal patterns with known structure. These characteristics arise from the physical properties and environmental conditions of the application. Standard robustness analysis formulations cannot take advantage of prior knowledge and only provide results for the full set of possible disturbances. Furthermore, analyzing a finite time horizon is usually sufficient as performance requirements are often defined over specific time intervals. Typical aerospace applications include space launchers with focus on structural loads during ascent (\cite{Biertuempfel2023}); missile engagements where the final deviation from the target is critical (\cite{Buch2021}); the response of flexible aircraft to wind gusts (\cite{Iannelli2019}); or auto-land scenarios of aircraft  (\cite{Biertuempfel2022}).  
	
	% V3
	%In aerospace applications, typical examples include missile engagements, where the final deviation from the target is critical (\cite{Buch2021}); space launchers, which focus on structural loads during a specific ascent phase (\cite{Biertuempfel2023}); or the response of flexible aircraft to wind gusts (\cite{Iannelli2019}).
	
	%Typical examples within the aerospace industry include missile engagements \cite{Buch2021}, launch vehicle control \cite{Biertuempfel2023} or unmanned aerial vehicle (UAV) missions \cite{Palframan2019}.  
	
	This paper proposes a novel formulation for robustness analysis of finite-horizon uncertain systems in the presence of mixed disturbances, i.e., simultaneously considering arbitrary $L_2$-bounded inputs and signals with partially known characteristics. Such signals could be, e.g., constant biases, bounded arbitrary time-varying behavior or harmonic excitation as considered in \cite{Cheah2024}. In the presented work, the known external disturbances are expressed as internal parameter variations driven by an artificial state. This state has an unknown initial value and follows from a non-standard system augmentation. The variation is then bounded using appropriate signal integral quadratic constraints (IQCs), which are particularly suited to recover various signal characteristics. Posing the problem inside the IQC framework also facilitates the incorporation of model uncertainties and nonlinearities in the same analysis. The underlying theorem is derived for the analysis of linear time-varying (LTV) systems with uncertain initial conditions over a finite time horizon.
	
	Applying IQCs to stability and robustness analyses is an established approach introduced by \cite{Megretski1997}. Their work derives a frequency-domain stability theorem for uncertain linear time-invariant (LTI) systems based on linear matrix inequalities. 
	%Their work derives a stability theorem in the frequency domain applicable to uncertain linear time-invariant (LTI) systems based on linear matrix inequalities. 
	Alternatively, time-domain formulations of IQCs, see, e.g., \cite{Joensson1996}, essentially overcome the limitations of LTI systems.  
	%Subsequent works such as \cite{Seiler2015} and \cite{Veenman2016} include time-domain formulations of IQCs, essentially overcoming the limitation of LTI systems. An introduction is provided in Section~\ref{sec:Back}. 
	They have been applied to the robustness analysis of parameter-varying systems, e.g., in \cite{Pfifer2016}, time-periodic systems in \cite{Ossmann2019} as well as time-varying systems. The analysis of (uncertain) time-varying systems was conducted in various literature, e.g., in \cite{Cantoni2013} IQCs and a gap metric provide a stability statement; \cite{Seiler2019} discusses bounds of the reachable set and robust induced gains for uncertain finite-horizon systems; and \cite{Farhood2024} addresses uncertain initial conditions. In \cite{Biertuempfel2023_unhc} the robustness of nonlinear systems along uncertain trajectories is analyzed, using a similar system augmentation step as in the presented paper. 
	
	The IQC framework has also been applied to cover the effects of specific shapes of external signals, often referred to as signal IQCs. In the works of \cite{Jonsson2003} or \cite{AyalaCuevas2019} the impact of oscillations on the system performance is discussed. An evaluation of a flight controller for restricted disturbances is conducted in \cite{Palframan2019} using linear fractional transformation. %where a linear fractional transformation is used for the problem formulation. 
	In \cite{Fry2021} signal IQCs are applied to cover noise in a performance setting, but their approach remains limited to the induced input-output behavior. This shortcoming is addressed with the approach presented in this work.

	%% Old  
	%Applying IQCs to stability and robustness analyses is an established approach introduced by \cite{Megretski1997}. Their work derives a stability theorem in the frequency domain applicable to linear time-invariant (LTI) systems based on linear matrix inequalities and contains an extensive collection of signal IQCs. In subsequent work such as \cite{Veenman2016}, time-domain formulations of signal IQCs have been presented, essentially enabling their use outside the scope of LTI systems. The time-domain IQC where applied to the robustness analysis of parameter-varying systems, e.g. in \cite{Pfifer2016}, and later applied to the specific cases of time-periodic (\cite{Ossmann2019}) and time-varying systems. The analysis of (uncertain) time-varying system was conducted in various literature, e.g. \cite{Cantoni2013} or \cite{Farhood2024}, with the latter also accounting for uncertain initial conditions. In \cite{Seiler2019} the reachable set of states for an uncertain LTV system under $L_2$-bounded exogenous disturbance is considered and evaluated by the induced $L_2$-norm over a finite time-horizon. For a similar problem formulation, the worst-case gain is provided in \cite{Biertuempfel2023}. 
	
	This paper introduces a novel unified framework that accommodates arbitrary norm-bounded and partially known input signals to provide more realistic analysis results in Section~\ref{sec:SigIQC}. It is applied to the path following problem of an unmanned aerial vehicle (UAV) described in \cite{Bertran2025} to evaluate the tracking performance in Section~\ref{sec:Exmp}. The UAV is subject to model uncertainties and affected by external wind disturbances with known characteristics, while tracking a norm bounded reference command. The results are compared with a classical IQC finite horizon robustness analysis to demonstrate the validity of the proposed approach.
	
	%%%%%%%%%%%%%%%%%%%%
	%%% ------------- %%
	%%% -- Section -- %%
	%%% ------------- %%
	%%%%%%%%%%%%%%%%%%%%
	\section{Uncertain Systems Over Finite Horizons}\label{sec:Back}
	
	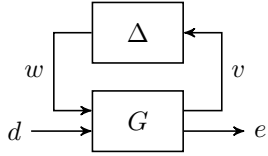
\begin{figure}[h!]
		\centering \begin{tikzpicture}[blockdiag, auto]

% Blocks
\node[block, minimum width=1.2cm](G){$G$};
\node[block, above = of G, minimum width = 1.2cm, yshift = -10](Delta){$\Delta$}; 

% Input and output of System 
\draw[<-] ($(G.north west)!0.66!(G.south west)$) -- +(-.8cm, 0) node[left]{$d$};
\draw[->] ($(G.north east)!0.66!(G.south east)$) -- +(.8cm, 0) node[right]{$e$};

% Connection to Delta Block
\draw[->] ($(G.north east)!0.33!(G.south east)$) -- +(.5cm, 0) |-  (Delta.east) node[right, pos = 0.25]{$v$};
\draw[<-] ($(G.north west)!0.33!(G.south west)$) -- +(-.5cm, 0) |-  (Delta.west) node[left, pos = 0.25]{$w$};

\end{tikzpicture}  
		\caption{Interconnection of a nominal LTV system $G$ and perturbation $\Delta$.}
		\label{blk:upperLFT} 
	\end{figure}
	
	An uncertain, finite-horizon linear and possibly time-varying system $\mathcal{F}_u(G, \Delta)$ is defined by the feedback interconnection of a nominal system $G$ and the perturbation $\Delta$ as pictured in Fig.~\ref{blk:upperLFT}.
	The nominal linear system $G$ is:
	\begin{equation}\label{eq:G}
		\left[\begin{array}{c} \dot{x} \\\hline v\\ e \end{array}\right] = 
		\left[\begin{array}{c|cc} A & B_w & B_d \\\hline C_v  & D_{vw} & D_{vd} \\ 
			C_e & D_{ew} & D_{ed} \end{array}\right]
		\left[\begin{array}{c} x \\\hline w \\ d \end{array}\right],
	\end{equation}
	where $x(t) \in \R^{n_{x}}$, $d(t) \in \R^{n_d}$, and $e(t) \in \R^{n_e}$ denote the state, disturbance input, and performance output vector at time $t$, respectively. The connection between the system $G$ and perturbation $\Delta$ is provided by the vectors $w(t)\in \mathbb{R}^{n_w}$ and $v(t)\in\mathbb{R}^{n_v}$, i.e. $w=\Delta(v)$.
	The matrices $A$, $B$, $C$, and $D$ are piecewise continuous locally bounded matrix-valued functions of time with dimensions corresponding to the multiplied vectors. The explicit time dependence is mostly omitted in this paper for brevity and will be clear from context. The uncertainty $\Delta: L_2^{n_v}[0,T] \rightarrow L_2^{n_w}[0,T]$ is a bounded and causal operator.
	It can describe nonlinearities like saturation, infinite-dimensional operators such as time delays, as well as dynamic and real parametric uncertainties. 
	
	In this paper, the input-output behavior of the perturbation $\Delta$ is bounded with time-domain IQCs. A time-domain IQC is defined by a stable filter $\Psi \in \RH^{n_z \times (n_v + n_w)}$ with the output $z\in\mathbb{R}^{n_z}$ and a real symmetric matrix $M = M^T \in \mathbb{R}^{n_z \times n_z}$ as described in \cite{Pfifer2016}. % $\RH^{n_z \times (n_v + n_w)}$ denotes a set of $z\times (n_v+n_w)$ matrices whose elements are rational functions with real coefficients that are in the closed right half of the complex plane. 
	If the output $z$ of the filter $\Psi$ satisfies the quadratic time-constraint
	\begin{equation}
		\label{eq:iqctd1}
		\int_0^T z(t)^T M z(t) \, dt \ge 0
	\end{equation}
	for all $v \in L_2[0,T]$ and $w=\Delta(v)$ over the interval $[0,T]$, the uncertainty $\Delta$ satisfies the IQC defined by $M$ and $\Psi$. This will be indicated by the short notation $\Delta \in IQC(\Psi,M)$. %A valid time domain IQC for a set $\mathcal{S}$ can be defined such that $\mathcal{S}\subseteq IQC(\Psi, M)$. 
	%In general, the time-domain constraint (\ref{eq:iqctd1}) for a given factorization $\Pi = \Psi^\sim M \Psi$ only holds over an infinite time horizon, i.e. $T=\infty$. \cite{Megretski2010} shows that under mild technical assumptions a large class of IQC multipliers possesses an equivalent finite horizon time-domain expression (\ref{eq:iqctd1}) .
	
	%An example for a SISO LTI dynamic uncertainty, commonly used to account for unmodeled dynamics, is given in Example \ref{exmp:LTI}:
	%%%%%%%
	%\begin{exmp}\label{exmp:LTI}
	%Let $\Delta$ be a LTI dynamic uncertainty, with $\Delta \in \RH$ and $\norm{\Delta}_{\infty} \le b \in \R$. A valid time-domain IQC for $\Delta$
	%is defined by $\Psi = \bsmtx b\psi_\nu \otimes I_{n_v} & 0 \\ 0 & \psi_\nu \otimes I_{n_v} \esmtx$ and $\mathcal{M} := \{ M = \bsmtx X\otimes I_{n_v} & 0 \\ 0 & -X\otimes I_{n_v}\esmtx : X = X^T > 0 \in \R^{(\nu+1)\times (\nu+1)}\}$.
	%A typical choice for $\psi_\nu \in \RH^{(\nu+1) \times 1}$ is:
	%\begin{equation}\label{eq:psi_nu}
	%\psi_\nu = \bmtx 1 & \frac{s+\rho}{s-\rho}& \dots & \frac{(s+\rho)^\nu}{{(s-\rho)}^\nu} \emtx^T\, , \, \rho < 0\, , \, \nu \in \mathbb{N}_0.
	%\end{equation}
	%\end{exmp}
	%In the example above, $\otimes$ denotes the Kronecker product. 

	From the worst-case analysis of nominal LTV systems in \cite{Green1994} and the finite-horizon time-domain IQC formulation of the perturbation $\Delta$, a worst-case gain condition can be derived as shown in \cite{Biertuempfel2018} and  \cite{Seiler2019}. It provides a guaranteed upper bound on the input-output behavior of an uncertain LTV system over the considered finite analysis horizon. Given a perturbation satisfying an IQC represented by $(\Psi,M)$, i.e. $\Delta \in IQC(\Psi,M)$, the interconnection $\mathcal{F}_u(G, \Delta)$ can be extended by the IQC filter $\Psi$. This procedure is illustrated in Fig.~\ref{fig:fbic_LTV_IC}.
	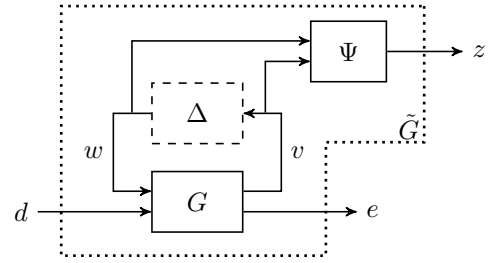
\begin{figure}[h!]
		\centering \begin{tikzpicture}[blockdiag, auto]

% Blocks
\node[block, minimum width=1.2cm] (G) {$G$};
\node[block, above=of G, minimum width=1.2cm, dashed, yshift = -10pt](Delta) {$\Delta$};
\node[block, above= of Delta, minimum width=1cm, xshift = 2cm, yshift = -20pt](Psi){$\Psi$};

% Connections for Psi  
\node[connector, left = of Delta, xshift = 13](ConW){}; 
\node[connector, right = of Delta, xshift = -12](ConV){}; 

% Input and output of System 
\draw[<-] ($(G.north west)!0.66!(G.south west)$) -- +(-1.5cm, 0) node[left]{$d$};
\draw[->] ($(G.north east)!0.66!(G.south east)$) -- +(1.5cm, 0) node[right]{$e$};

% Connection to Delta 
\draw[->] ($(G.north east)!0.33!(G.south east)$) -- +(.5cm, 0) |-  (Delta.east) node[right, pos = 0.25]{$v$};
\draw[<-] ($(G.north west)!0.33!(G.south west)$) -- +(-.5cm, 0) |-  (Delta.west) node[left, pos = 0.25]{$w$};

% Connection to Psi
\draw[->] (ConW) |- ($(Psi.north west)!0.33!(Psi.south west)$); 
\draw[->] (ConV) |- ($(Psi.north west)!0.66!(Psi.south west)$); 
\draw[->] (Psi.east) -- +(1cm, 0) node[right]{$z$};

\draw[dotted, line width=1pt] (Psi) ++ (1.cm,0.6cm) |- ++ (-1.3cm, -1.8cm)node[left, yshift = 6pt, xshift = 2pt, pos=0.5]{$\tilde{G}$} |- ++(-3.5cm,-1.5cm) |- ++(4.8cm, 3.3cm);

\end{tikzpicture}  
		\caption{Extended state-space system $\tilde{G}$.}
		\label{fig:fbic_LTV_IC}
	\end{figure} 
	The dynamics of this interconnection is referred to as $\tilde{G}$ and defined by
	
	\begin{equation}
		\label{eq:Pext}
		\left[\begin{array}{c} \dot{\tilde{x}} \\\hline z\\ e \end{array}\right] = 
		\left[\begin{array}{c|cc} \tilde{A} & B_1 & B_2 \\\hline C_1  & D_{11} & D_{12} \\ 
			C_2 & D_{21} & D_{22} \end{array}\right]
		\left[\begin{array}{c} \tilde{x} \\\hline w \\ d \end{array}\right].
	\end{equation}
	
	In (\ref{eq:Pext}), $\tilde{x}(t)$ contains the states of $G$ and $\Psi$. By enforcing the time-domain inequality (\ref{eq:iqctd1}) on the filter output $z$, the explicit representation of the uncertainty $w = \Delta(v)$ is replaced. The finite horizon worst-case induced $L_2[0,T]$-gain is defined as 
	\begin{equation}
		\label{eq:E2PWC}
		\begin{gathered}
			\| \mathcal{F}_u(G,\Delta) \|_{2[0,T]} := \sup_{ \Delta \in \textrm{IQC}(\Psi ,M)} \sup_{\substack{d \in
					L_2[0,T]\\ d \neq 0, x(0) = 0}} \frac{ \norm{e}_{2[0,T]}}{\norm{d}_{2[0,T]}},
		\end{gathered}
	\end{equation}
	with, e.g., $\norm{d}_{2[0,T]} = [\int_{0}^{T}d(t)^Td(t)\mathrm{d}t]^{\frac{1}{2}}$.
	This gain can be interpreted as a bound on the energy amplification from the disturbance input to the performance output over the time horizon $[0,T]$ and all $\Delta \in IQC(\Psi, M)$.
	A dissipation inequality can be stated to upper bound the worst-case induced $L_2[0,T]$-gain of the interconnection $\mathcal{F}_u(G, \Delta)$.
	It is based on the extended LTV system $\tilde{G}$ given by (\ref{eq:Pext}) and the finite horizon time-domain IQC formulation in (\ref{eq:iqctd1}). %For a more thoroughly description, the reader is referred to the work of \cite{Seiler2019} and \cite{Biertuempfel2023}.
	This dissipation inequality is rearranged as an equivalent Riccati differential equation (RDE) formulation in the following Theorem \ref{thm:E2P}:% {\color{red}(FB: Write as LMI?)}
	\begin{thm}\label{thm:E2P}
		Let $\mathcal{F}_u(G, \Delta)$ be well-posed $\forall \, \Delta \in IQC(\Psi, M)$, 
		then $\|F_u(G,\Delta)\|_{2[0,T]} < \gamma$ if there exist a
		continuously differentiable symmetric $P: \R^+_0 \rightarrow \R^{n_x \times n_x}$ such that
		\begin{equation}
			\label{eq:Pcondition_L2}
			P(T) = 0,
		\end{equation}
		\begin{equation}
			\label{eq:classicRDE}
			\begin{split}
				\dot P = \hat{Q} + P\hat{A}+ \hat{A}^TP-P\hat{S}P \qquad \forall t \in [0,T]
			\end{split}
		\end{equation}
		and
		\begin{equation}
			\label{Rinv}
			\begin{split}
				\hat{R} = \bsmtx D_{11}^TMD_{11} + D_{21}^TD_{21} & D_{11}^TMD_{12} + D_{21}^TD_{22}\\
				D_{12}^TMD_{11} + D_{22}^TD_{21}   & D_{12}^TMD_{12} + D_{22}^TD_{22} - \gamma^2I_{n_d}\esmtx <0,
			\end{split}
		\end{equation}
		with
		\begin{equation}
			\label{eq:Atilde}
			\begin{split}
				\hat{A} = \bsmtx B_1 & B_2\esmtx \hat{R}^{-1} \bsmtx (C_1^TMD_{11} +C_2^TD_{21})^T \\ (C_1^TMD_{12} + C_2^TD_{22})^T \esmtx - \tilde{A},
			\end{split}
		\end{equation}
		\begin{equation}
			\label{eq:S}
			\begin{split}
				\hat{S} = -\bsmtx B_1 & B_2 \esmtx \hat{R}^{-1} \bsmtx B_1^T \\ B_2^T \esmtx,
			\end{split}
		\end{equation}
		and
		\begin{equation}
			\label{eq:Q}
			\begin{split}
				\hat{Q} =& - C_1^TMC_1 - C_2^TC_2 \\
				&+\bsmtx (C_1^TMD_{11} +C_2^TD_{21})^T \\ (C_1^TMD_{12} + C_2^TD_{22})^T \esmtx^T \hat{R}^{-1} 
				\bsmtx (C_1^TMD_{11} +C_2^TD_{21})^T \\ (C_1^TMD_{12} + C_2^TD_{22})^T\esmtx.
			\end{split}
		\end{equation}
	\end{thm}
	\begin{pf}
		The proof is provided in  \cite{Biertuempfel2018}.
	\end{pf}
	For a more thorough description, the reader is referred to the work of \cite{Seiler2019} and \cite{Biertuempfel2023}. Therein, efficient approaches for computing an upper bound on the worst-case gain are provided. \\
	\emph{Remark:} Theorem \ref{thm:E2P} is limited to the analysis of arbitrary $L_2[0,T]$ worst-case disturbances maximizing the induced $L_2[0,T]$-gain. Corollaries of this theorem can cover other quadratic performance metrics but are limited to induced input-output gains. It is also possible to include different performance metrics in one analysis using dedicated performance blocks, see, e.g., \cite{Fry2021}. However, these blocks are limited to specific input-output pairs, i.e., there is a one-to-one correspondence between a specific disturbance and performance output. Hence, mixed disturbance signals affecting one or more performance outputs cannot be covered. For example, assume that Theorem \ref{thm:E2P} is applied to a problem with two disturbances affecting a single output. One is well represented by an arbitrarily shaped norm-bounded signal. The other has a known maximum amplitude but otherwise arbitrary shape. However, the analysis will simply provide an upper bound for the worst-case combination of both signals. The maximum amplitude of the second disturbance cannot be enforced, which yields inaccurate results.

	%%%%%%%%%%%%%%%%%%%%
	%%% ------------- %%
	%%% -- Section -- %%
	%%% ------------- %%
	%%%%%%%%%%%%%%%%%%%%
	\section{Analysis of Mixed-Disturbances Using Signal IQCs}\label{sec:SigIQC} 
	\subsection{Analysis Condition}
	This section derives a novel formulation for robustness analysis that allows to account for mixed disturbances. In particular, the analysis considers arbitrary norm-bounded worst-case signals and signals with (partially) a priori known properties affecting the same performance output. The latter will be incorporated using \textit{signal IQCs} effectively decoupling them from the former.
	Consider a finite-horizon linear system with the external disturbances~$d$~and~$\delta$
	\begin{equation}
		\label{eqn:stateSpace_distSplit}
		\left[\begin{array}{c} \dot{x} \\\hline e \end{array}\right] = 
		\left[\begin{array}{c|cc} A  & B_{d} & B_{\delta} \\\hline
			C & D_d & D_\delta \end{array}\right]
		\left[\begin{array}{c} x \\\hline d \\ \delta \end{array}\right], 
	\end{equation}
	where $d \in L^{n_d}_{2}[0,T]$ are norm-bounded arbitrary disturbance signals and $\delta(t)\in \R^{n_\delta}$ are disturbances with partially known characteristics.
	The disturbance signals $\delta$ are multiplied by a constant driving term with value $1$. Extending the state vector with this driving term pushes the term $B_\delta \delta$ into the state matrix and $D_\delta \delta$ into the output matrix. The resulting state-space representation is:
	\begin{equation}\label{eq:Gext}
		\left[\begin{array}{c} \dot{x} \\ 0 \\\hline  e \end{array}\right] = 
		\left[\begin{array}{cc|c} A & B_{\delta} \delta  & B_{d} \\ 0  & 0 & 0 \\\hline  
			C & D_\delta \delta  & D_d \end{array}\right]
		\left[\begin{array}{c} x \\ 1 \\\hline  d \end{array}\right].
	\end{equation}
	This system is non-standard but identical to the system~\eqref{eqn:stateSpace_distSplit}.
	In the next step, the driving term is replaced with an artificial state $x_\delta :=1$ providing the augmented system:
	\begin{equation}\label{eq:Gpseudo}
		\left[\begin{array}{c} \dot{x} \\ 0 \\\hline  e \end{array}\right] = 
		\left[\begin{array}{cc|c} A & B_{\delta} \delta  & B_{d} \\ 0  & 0 & 0 \\\hline  
			C & D_\delta  & D_d\end{array}\right]
		\left[\begin{array}{c} x \\ x_\delta \\\hline  d \end{array}\right].
	\end{equation}
	As the signals $\delta$ have known properties, they can without loss of generality be confined to a set of expected disturbances $\mathcal{D}\subseteq \R^{n_\delta}$. By restricting the disturbance signals $\delta$ to this set, they can be treated equivalently as perturbations $\Delta_\text{S}:= \mathrm{diag}(\delta_1,\, \delta_2,\dots,\, \delta_{n_\delta})$. The external disturbances thus become internal signals described by ${\Delta_\text{S}}$ which are driven by a single artificial state. In other words, one artificial state is sufficient to persistently excite multiple internal signals. An uncertain finite-horizon system comparable to \eqref{eq:G} can be recovered which provides the system $H$: 
	\begin{equation}\label{eq:Gfinal}
		\begin{split}
			\left[\begin{array}{c} \dot{x} \\ 0 \\\hline v_\delta \\  e \end{array}\right] &= 
			\left[\begin{array}{cc|cc} A & 0 & B_\delta & B_{d} \\ 0 & 0 & 0 & 0 \\\hline 0 & 1_{n_\delta} & 0 & 0 \\ 
				C & 0 & D_\delta & D_d \end{array}\right]
			\left[\begin{array}{c} x \\ x_\delta \\\hline w_\delta \\ d \end{array}\right]\\
			w_\delta &= {\Delta_\text{S}}(v_\delta), 
		\end{split}
	\end{equation}
	with $v_\delta(t) \in \mathbb{R}^{n_\delta} \coloneqq x_\delta$, $w_\delta(t) \in \mathbb{R}^{n_\delta}$, and ${1}_{n_\delta}$ the 1-vector of size $n_\delta$. 
	The input-output behavior of these internal signals can be upper bounded with a suitable signal IQC, such that $\Delta_\text{S} \in IQC(\Psi, M)$. Extending the system~\eqref{eq:Gfinal} with the IQC filter $\Psi$ yields the state-space system required for an IQC-based analysis:
	\begin{equation}
		\label{eq:PextFinal}
		\left[\begin{array}{c} \dot{\bar{x}} \\\hline z\\ e \end{array}\right] = 
		\left[\begin{array}{c|cc} \mathcal{A} & \mathcal{B}_1 & \mathcal{B}_2 \\\hline \mathcal{C}_1  & \mathcal{D}_{11} & \mathcal{D}_{12} \\ 
			\mathcal{C}_2 & \mathcal{D}_{21} & \mathcal{D}_{22} \end{array}\right]
		\left[\begin{array}{c} \bar{x} \\\hline w_\delta \\ d \end{array}\right].
	\end{equation}
	This extended system is closely related to $\tilde{G}$ in~\eqref{eq:Pext}, but its extended state vector $\bar{x}$ also contains the artificial state $x_\delta := 1$, i.e., $\bar{x} = [\tilde{x}^T, x_\delta]^T$.
	Theorem \ref{thm:E2P} provides an analysis condition to bound the worst-case induced $L_2[0,T]$ gain in \eqref{eq:E2PWC} for zero initial conditions. Thus, it cannot be applied directly to this problem. Extensions of the (nominal) Bounded Real Lemma (BRL) to uncertain initial conditions as in \cite{Khargonekar1991} provide corresponding bounds for induced gains. These have numerators of the form $\sqrt{\norm{d}_{2[0,T]}^2 + x^T(0)Nx(0)}$. Here, the matrix $N~>~0~\in~\R^{n_x\times n_x}$ quantifies the \textit{relative importance} of the initial condition with respect to the external disturbances. Thus, $x(0)$ cannot be restricted to a specific set as required in~\eqref{eq:Gfinal}. 
	
	To derive a suitable analysis condition, we exploit the fact that for zero initial conditions any disturbance can be scaled such that $\norm{d}_{2[0,T]} = 1$. In this case, Theorem \ref{thm:E2P} provides an upper bound on the performance output $\norm{e}_{2[0,T]}\le \gamma$. The following theorem provides an analysis condition to calculate an upper bound on the worst-case $L_2[0,T]$-norm of the performance output $\norm{e}_{2[0,T]}$ which confines the artificial state $x_\delta$ to $\{ x_\delta\in \R : \abs{x_\delta}\le 1 \}$ and also enforces $\norm{d}_{2[0,T]}\le 1$.

	\begin{thm}\label{thm:UncInit} 
		Let $\mathcal{F}_u(H,\Delta_\text{S})$ be well-posed for all $\Delta_\text{S}\in IQC(\Psi, M)$. Then the worst case $L_2[0,T]$ norm of the performance output $e$ is bounded by the scalar $\gamma > 0$ for
		$\norm{d}_{2[0,T]}\le1$ and $\abs{x_\delta(0)}\le 1$ if there exist scalars $\alpha_1>0$ and $\alpha_2>0$, and a continuously differentiable symmetric matrix function $P(t) = \begin{bmatrix}P_{11} & P_{12} \\ P_{12}^T & P_{22} \end{bmatrix}$ such that 
		\begin{equation}\label{eq:LMI1}
			\begin{split}
				&\begin{bmatrix}
					\dot{P} + P\mathcal{A} + \mathcal{A}^TP & P\mathcal{B}_1 &  P\mathcal{B}_2 \\
					\mathcal{B}_1^TP & 0 & 0 \\ 
					\mathcal{B}_2^TP & 0 & -I
				\end{bmatrix} + \begin{bmatrix}
					\mathcal{C}_1^T \\ \mathcal{D}_{11} \\ \mathcal{D}_{12}^T
				\end{bmatrix} M \begin{bmatrix}
					\mathcal{C}_1^T \\ \mathcal{D}_{11}^T \\ \mathcal{D}_{12}^T
				\end{bmatrix}^T \\
				&+\alpha_2 \begin{bmatrix}
					\mathcal{C}_2^T \\ \mathcal{D}_{21} \\ \mathcal{D}_{22}^T
				\end{bmatrix}  \begin{bmatrix} \mathcal{C}_2^T \\ D_{21}^T \\ \mathcal{D}_{22}^T \end{bmatrix}^T 
				< 0 
			\end{split}
		\end{equation}
		and 
		\begin{equation}\label{eq:P22}
			\begin{bmatrix}
				P_{22}(0)-\alpha_1 & 0 & 0 \\ 
				0 & -P(T) & 0\\
				0 & 0 & \alpha_1 - \alpha_2 \gamma^2 + 1
			\end{bmatrix} \le 0 
		\end{equation}
	\end{thm} 
	\begin{pf}
		The proof relies on the definition of a time-varying quadratic storage function~$V(\bar{x},t) = \bar{x}(t)^T P(t)\bar{x}(t)$ and uses a standard dissipation argument. Perturbing the left-hand-side of~\eqref{eq:LMI1} with $(1-\epsilon)$ with $0<\epsilon \ll 1$, multiplying the left and right side with~$[\bar{x}^T, w_\delta^T, d^T]$ and~$[\bar{x}^T, w_\delta^T, d^T]^T$, respectively yields
		\begin{equation}\label{eq:pf1}
			\begin{split}
				\dot{V}(\bar{x},t) + z^TMz
				+ \alpha_2 e^Te- (1-\epsilon)d^Td \le 0.
			\end{split}
		\end{equation}
		Integrating this dissipation inequality from $0$ to $T$ and recalling that only $x_\delta(0) \neq 0$ provides
		\begin{equation}\label{eq:pf1}
			\begin{split}
				\bar{x}(T)^TP(T)\bar{x}(T) - x_\delta(0)^2P_{22}(0) + \int_0^Tz(t)^TMz(t) dt& \\
				+ \alpha_2 \norm{e}^2_{2[0,T]} - (1-\epsilon)\norm{d}^2_{2[0,T]} \le 0.&
			\end{split}
		\end{equation}
		Multiplying the left-hand-side of~\eqref{eq:P22} with $[x_\delta(0), \bar{x}(T)^T, 1]$ and $[x_\delta(0), \bar{x}(T)^T, 1]^T$, respectively results in
		\begin{equation}\label{eq:pf2}
			\begin{split}
				x_\delta(0)^2P_{22}(0) - \alpha_1 x_\delta(0)^2 -\bar{x}(T)^TP(T)\bar{x}(T)&\\
				+ \alpha_1 -\alpha_2 \gamma^2 +1 \le 0.&
			\end{split}
		\end{equation}
		As $\Delta_\text{S} \in IQC(\Psi,M)$, it follows that $\int_0^Tz(t)^TMz(t) dt \ge 0$ and~\eqref{eq:pf1} can be substituted into~\eqref{eq:pf2} which yields the inequality
		\begin{equation}\label{eq:pf3}
			1-(1-\epsilon)\norm{d}^2_{2[0,T]}+\alpha_2(\norm{e}^2_{2[0,T]} -\gamma^2) +\alpha_1(1-x_\delta(0)^2) \le 0.
		\end{equation}
		From the definitions $\norm{d}_{2[0,T]}\le 1$ and $\abs{x_\delta} \le 1$, it can be concluded that $\norm{e}_{2[0,T]}\le \gamma$.$\qed$
	\end{pf}
	Note that in the analysis only the artificial state $x_\delta$ has non-zero initial conditions. Accordingly, Theorem~\ref{thm:UncInit} partitions $P$ 
	into a $n_{\tilde{x}} \times n_{\tilde{x}}$ matrix $P_{11}$ (related to the state $\tilde{x}$) and a scalar $P_{22}$ (related to the artificial state $x_\delta$). The upper left entry in~\eqref{eq:P22} constrains the initial value $x_\delta(0)$.

	\subsection{Computational approach}
	
	A few steps are necessary to convert Theorem~\ref{thm:UncInit} into a computationally tractable problem.
	First, an infinite number of feasible IQCs exists to bound the behavior of a given signal. The most common
	solution to this problem found in literature, e.g., \cite{Veenman2016} or \cite{Pfifer2016}, uses a fixed IQC filter $\Psi$ and parameterizes $M$. In other words, $M$ is restricted to a feasible set $\mathcal{M}$ such that $\Delta_\text{S}\in IQC(\Psi, M)$.
	In \cite{Megretski1997} and \cite{Veenman2016} extensive catalogs of
	IQCs suitable to describe different types of signals are provided. Two examples for useful
	IQC parameterizations for signals are given below. The first example concerns a constant disturbance whose absolute value is known to be bounded.
	\begin{exmp}[Constant Disturbance]\label{ex:delta}
		Let $\Delta_\text{S} = \delta$ be a constant signal $\delta$, with $|\delta(t)| \le b \in \R$. A valid time-domain IQC for $\Delta_\text{S}$
		is defined by $\Psi = \text{diag}(\psi_\nu, \psi_\nu)$ and $\mathcal{M} := \{M = \bsmtx b^2X & Y \\ Y^T & -X\esmtx : X = X^T > 0 \in \R^{(\nu+1)\times (\nu+1)}, Y = -Y^T \in \R^{(\nu+1) \times (\nu+1)}\}$.
		A typical choice for $\psi_\nu \in \RH^{(\nu+1) \times 1}$ is
		$\psi_\nu = \bmtx 1 & \frac{s+\rho}{s-\rho}& \dots & \frac{(s+\rho)^\nu}{{(s-\rho)}^\nu} \emtx^T\, , \, \rho < 0\, , \, \nu \in \mathbb{N}_0$.
	\end{exmp}
%	\begin{exmp}[Constant Disturbance]\label{ex:delta}
%		Let $\Delta_\text{S} = \delta$ be a constant signal $\delta$, with $|\delta(t)| \le b \in \R$. A valid time-domain IQC for $\Delta_\text{S}$
%		is defined by $\Psi = \text{diag}(\psi_\nu, \psi_\nu)$ and $\mathcal{M}$ defined as the set of all matrices $M = \bsmtx b^2X & Y \\ Y^T & -X\esmtx$ such that $X\!=\!X^T\!>\!0 \in \R^{(\nu+1)\times (\nu+1)}$ and $Y\!=\!-Y^T \in \R^{(\nu+1) \times (\nu+1)}$ with $\nu \in \mathbb{N}_0$.
%		A typical choice for $\psi_\nu \in \RH^{(\nu+1) \times 1}$ is	$\psi_\nu = \bmtx 1 & \frac{s+\rho}{s-\rho}& \dots & \frac{(s+\rho)^\nu}{{(s-\rho)}^\nu} \emtx^T$ with $\rho < 0$.
%	\end{exmp}
%	
	This example may appear trivial, but constant external disturbances, e.g., biases, cannot be covered in
	standard BRL-based worst-case analyses. 
	The next example describes an arbitrarily time-varying unknown disturbance signal with bounded maximum amplitude. 
	\begin{exmp}[Time-Varying Disturbance]\label{exmp:tv}
		Let $\Delta_\text{S} = \delta$ be an arbitrarily time-varying signal, with $|\delta(t)| \le b \in \R\, \forall\, t\in [0,T]$. A valid time-domain IQC for $\Delta_\text{S}$ is defined by $\Psi = \bsmtx 1 & 0 \\ 0 & 1 \esmtx$ and $\mathcal{M} := \{ M = \bsmtx b^2X & 0 \\ 0 & -X\esmtx : X > 0 \in \R\}$.
	\end{exmp}

	%\begin{figure}
	%	\centering \input{figures/tikzBlk/Blk_AdvancedAnalysis.tikz} 
	%	\caption{Analysis Problem setup for mixed disturbances.}
	%	\label{blk:NewAnalysis}
	%\end{figure}
	%\begin{figure}
	%	\centering \input{figures/tikzBlk/Blk_AdvancedAnalysis_V2.tikz} 
	%	\caption{Analysis Problem setup for mixed disturbances.}
	%	\label{blk:NewAnalysisV2}
	%\end{figure}
	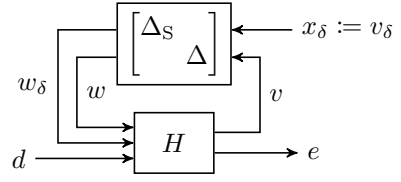
\begin{figure}
		\centering \begin{tikzpicture}[blockdiag, auto]

% Blocks 
\node[block](Plant){$H$}; 
\node[block, above = of Plant, yshift = -10pt](Delta){$\begin{bmatrix}\Delta_\text{S} & \phantom{0} \\  \phantom{0} & \Delta \end{bmatrix}$}; 

% Input to Plant 
\draw[<-] ($(Plant.north west)!0.75!(Plant.south west)$) -- +(-1.3cm, 0)node[left]{$d$};
\draw[<-] ($(Plant.north west)!0.25!(Plant.south west)$) -- +(-0.75cm, 0) |- ($(Delta.north west)!0.66!(Delta.south west)$)node[right, pos = 0.25]{$w$} ; 
\draw[<-] ($(Plant.north west)!0.5!(Plant.south west)$) -- +(-1cm, 0) |- ($(Delta.north west)!0.33!(Delta.south west)$)node[left, pos = 0.25]{$w_\delta$} ; 

% Output from Plant 
\draw[->] ($(Plant.north east)!0.66!(Plant.south east)$) -- +(1.1cm, 0)node[right]{$e$};

% Input to Delta
\draw[<-] ($(Delta.north east)!0.33!(Delta.south east)$) -- +(0.8cm, 0)node[right]{$x_\delta \coloneqq v_\delta$};
\draw[->] ($(Plant.north east)!0.33!(Plant.south east)$) -- +(0.6cm, 0) |- ($(Delta.north east)!0.66!(Delta.south east)$)node[right, pos = 0.25]{$v$};

%%%% old 

%% Blocks
%\node[block, minimum width=1.2cm] (G) {$G$};
%\node[block, above=of G, minimum width=1.2cm, dashed, yshift = -10pt](Delta) {$\Delta$};
%\node[block, above= of Delta, minimum width=1cm, xshift = 2cm, yshift = -20pt](Psi){$\Psi$};
%
%% Connections for Psi  
%\node[connector, left = of Delta, xshift = 13](ConW){}; 
%\node[connector, right = of Delta, xshift = -12](ConV){}; 
%
%% Input and output of System 
%\draw[<-] ($(G.north west)!0.66!(G.south west)$) -- +(-1.5cm, 0) node[left]{$d$};
%\draw[->] ($(G.north east)!0.66!(G.south east)$) -- +(1.5cm, 0) node[right]{$e$};
%
%% Connection to Delta 
%\draw[->] ($(G.north east)!0.33!(G.south east)$) -- +(.5cm, 0) |-  (Delta.east) node[right, pos = 0.25]{$v$};
%\draw[<-] ($(G.north west)!0.33!(G.south west)$) -- +(-.5cm, 0) |-  (Delta.west) node[left, pos = 0.25]{$w$};
%
%% Connection to Psi
%\draw[->] (ConW) |- ($(Psi.north west)!0.33!(Psi.south west)$); 
%\draw[->] (ConV) |- ($(Psi.north west)!0.66!(Psi.south west)$); 
%\draw[->] (Psi.east) -- +(1cm, 0) node[right]{$z$};
%
%\draw[dotted, line width=1pt] (Psi) ++ (1.cm,0.6cm) |- ++ (-1.3cm, -1.8cm)node[left, yshift = 5pt, pos=0.5]{$H$} |- ++(-3.5cm,-1.5cm) |- ++(4.8cm, 3.3cm);

\end{tikzpicture}  
		\caption{Analysis setup concept with system uncertainties}
		\label{blk:NewAnalysisV3}
	\end{figure}

	After choosing a suitable IQC parmeterization, the remainder of the theorem is addressed.
	The LMI (\ref{eq:LMI1}) has to hold for all $t\in [0,T]$, which poses an infinite number of constraints.
	The most common approach to render the analysis condition computationally feasible is to enforce the inequality condition only on a finite set of $i$ grid points $t_i \in [0, T]$, see, e.g., \cite{Pfifer2016}.
	The decision variables in conditions (\ref{eq:LMI1}) and (\ref{eq:P22}) are the matrix function $P$, the IQC parameterization $M\in \mathcal{M}$, and the positive scalars $\alpha_1$ and $\alpha_2$. The time-varying function $P$  must also be constrained to a finite dimensional subspace to pose a computationally tractable problem.
	Most commonly, these are expressed as linear combinations, e.g.  $P(t)=\sum_{i=0}^{N_\text{b}}t^i P_i$, $i = 0,1, ..., N_\text{b}$, where $N_\text{b}$ denotes the order of the basis function. Thus, the coefficients $P_i$ become the decision variables. More sophisticated basis functions like cubic splines can also be employed, as in \cite{Seiler2019}. 
	Feasible coefficients can then be calculated using a semidefinite program with the constraints (\ref{eq:LMI1}) and (\ref{eq:P22}) for a given $\gamma$. Bisecting over $\gamma$ calculates the minimal feasible upper bound on $\norm{e}_{2[0,T]}$.
	
	Additional uncertainties $\Delta$ in the system dynamics can be readily included by following the explanations in Section~\ref{sec:Back}. Doing so results in an extended uncertainty block $\bold{\Delta} := \text{diag}(\Delta_\text{S}, \Delta)$ in Theorem~\ref{thm:UncInit}. Figure~\ref{blk:NewAnalysisV3} provides an graphical interpretation of the corresponding analysis setup for the system $H$ in (\ref{eq:Gfinal}) including an additional uncertainty $w = \Delta(v)$ analog to (\ref{eq:G}).

	%\begin{exmp}\label% this is a full block LTI uncertainty
	%will be changed....Let $\Delta$ be a full-block LTI dynamic uncertainty, with $\Delta \in \RH$ and $\norm{\Delta}_{\infty} \le b \in \R$. A valid time-domain IQC for $\Delta$
	%is defined by $\Psi = \bsmtx b\psi_\nu \otimes I_{n_v} & 0 \\ 0 & \psi_\nu \otimes I_{n_v} \esmtx$ and $\mathcal{M} := \{ M = \bsmtx X\otimes I_{n_v} & 0 \\ 0 & -X\otimes I_{n_v}\esmtx : X = X^T > 0 \in \R^{(\nu+1)\times (\nu+1)}\}$.
	%A typical choice for $\psi_\nu \in \RH^{(\nu+1) \times 1}$ is again (\ref{eq:psi_nu}).
	%\end{exmp}
	
	%\textit{Remark:\,}Considering additional uncertainties $\Delta$ in the system dynamics is straight forward and follows from the explanations in section~\ref{sec:Back}. Doing so results in an extended uncertainty block $\bold{\Delta} := \text{diag}(\tilde{\Delta}, \Delta)$.

	%%%%%%%%%%%%%%%%%%%%%%%%%%%%%%%
	%% ------------------------- %%
	%% -- Example - UAV-Path  -- %%
	%% ------------------------- %%
	%%%%%%%%%%%%%%%%%%%%%%%%%%%%%%%
	
	\section{Numerical Example: UAV Path-Following Analysis} \label{sec:Exmp}
	
	This section presents a robustness analysis of a load factor tracker of an UAV under wind turbulence and system uncertainties. The controller was developed as part of the integrated guidance and robust control architecture for following energy optimal flight paths in urban environments (\cite{Bertran2025}). Here, robustness is crucial, as the trajectory is in the proximity of buildings to maximize the energy savings. %The considered UAV is the \emph{Urban Condor}, an electric Sig Kadet LT-40 modified for larger payloads by the Chair of Flight Mechanics and Control at TU Dresden. % A aerodynamic derivatives describing the UAV were identified based on wind tunnel tests conducted at the Chair and presented in the work of \cite{Wisbacher2025}.
	
	%\textit{UrbanCondor} following a predefined trajectory in an urban environment. The \textit{UrbanCondor} is an electric Sig Kadet LT-40 modified by the Chair of Flight Mechanics and Control at TU Dresden. The controller was developed as part of the integrated guidance and robust control architecture in \cite{Bertran2025} to follow energy optimal flight paths provided by \cite{Rienecker2024}. For this application the robustness is crucial, as the trajectory passes in close proximity to buildings to maximize the energy savings. 
	
	%This section presents the worst-case performance analysis of the UAV \textit{UrbanCondor} following a predefined trajectory in an urban environment. The \textit{UrbanCondor} is an electric Sig Kadet LT-40 modified by Chair of Flight Mechanics and Control at TU Dresden for operations in urban environments. The flight path is an energy optimal trajectory calculated in \cite{Rienecker2024}. The guidance and control algorithm is described in \cite{Bertran2025}. 
	%The goal of the analysis here is to assess the degradation of the tracking performance under wind disturbance and uncertainty in the aircraft dynamics. To keep the example concise, only the inner control loop tracking the load factor $n_z$ in the longitudinal dynamics is considered.
	
	\newcommand{\UAV}{\mathrm{U\!AV}}
	\newcommand{\wind}{w}%\mathrm{w}}
\newcommand{\cmd}{\mathrm{cmd}}
\newcommand{\act}{\mathrm{Act}}
\subsection{Augmented Aircraft Model}

The nonlinear longitudinal dynamics of the UAV are linearized for a velocity of $17$\,m/s and an altitude of $20$\,m. The resulting state-space representation is defined by the state vector $x~=~[U\;W\;\theta\;q]^T$, with the pitch angle $\theta$, pitch rate $q$ and the velocities $U$ and $W$ in the aircraft's body-fixed frame pointing in the direction of the nose and downwards, respectively. The system's inputs are the horizontal wind component $\delta_\wind$ and the elevator deflection $\delta_\text{e}$. The expected wind disturbance is based on measurements taken in the city of Dresden, Germany. Specifically, it corresponds to the average of the mean hourly wind speeds during the year's windiest day. The maximum amplitude of the wind disturbance is bounded by $|\delta_\wind|\leq b_{\delta_\wind}=5$\,m/s. The wind disturbance is included in the robustness analysis following the steps in Section~\ref{sec:SigIQC} which leads to an augmented state-space system with the artificial state $x_{\delta_\wind}$ acting as a driving term. The augmented system $G_{\UAV}$ is:
\begin{small}
	\begin{equation}\label{eq:PitchExt}
		\begin{split}
			\begin{bmatrix} \dot{U}\\  \dot{W} \\ \dot{\theta} \\ \dot{q}\\ 0 \end{bmatrix}\!\!\! &= \!\!\!
			\begin{bmatrix}
				-0.19 &  0.32  &-9.8 & 1.02  & 0.14\,\delta_\wind \\
				-1.58 & -6.32  & 0.67& 15.1  & 0.63\,\delta_\wind\\  
				0 & 0    & 0     & 1     & 0 \\
				-0.25     & -3.75& 0 & -12.5 & -0.01\,\delta_\wind \\ 
				0 & 0    & 0     & 0     & 0
			\end{bmatrix} \!\!\!
			\begin{bmatrix} U \\ W \\  \theta \\ q \\ x_{\delta_\wind}	\end{bmatrix}\!\!\! + \!\!\!
			\begin{bmatrix} -0.013 \\ 0  \\ -0.19 \\ -0.98 \\ 0 \end{bmatrix}\!\! \delta_e \\ 
			n_z \!\! &= \!\!
			\begin{bmatrix} 0.16 & 0.64 & 0 & 0.19 & -0.06\,\delta_\wind	\end{bmatrix}
			\begin{bmatrix} U \\ W \\ \theta \\ q \\ x_{\delta_\wind}	\end{bmatrix} + 0.0194\, \delta_e, 
		\end{split}
	\end{equation}
\end{small}
%Structure from matlab
%\begin{small}
%	\begin{equation}\label{eq:PitchExt}
	%		\begin{split}
		%			\begin{bmatrix} \dot{U}\\  \dot{\theta} \\ \dot{W} \\ \dot{q}\\ 0 \end{bmatrix}\!\!\! &= \!\!\!
		%			\begin{bmatrix}
			%				-0.19 & -9.8 & 0.32  & 1.02  & 0.14\,\delta_\wind \\ 
			%				0 & 0    & 0     & 1     & 0 \\ 
			%				-1.58 & 0.67 & -6.32 & 15.1  & 0.63\,\delta_\wind\\ 
			%				-0.25 & 0    & -3.75 & -12.5 & -0.01\,\delta_\wind \\ 
			%				0 & 0    & 0     & 0     & 0
			%			\end{bmatrix} \!\!\!
		%			\begin{bmatrix} U \\  \theta \\ W \\ q \\ x_{\delta_\wind}	\end{bmatrix}\!\!\! + \!\!\!
		%			\begin{bmatrix} -0.013 \\ 0  \\ -0.19 \\ -0.98 \\ 0 \end{bmatrix}\!\! \delta_e \\ 
		%			\begin{bmatrix} n_z \end{bmatrix}\!\! &= \!\!
		%			\begin{bmatrix} 0.16 & 0 & 0.64 & 0.19 & -0.06\,\delta_\wind	\end{bmatrix}
		%			\begin{bmatrix}U \\  \theta \\ W \\ q \\ x_{\delta_\wind}	\end{bmatrix} + 0.0194\, \delta_e, 
		%		\end{split}
	%	\end{equation}
%\end{small}
with $n_z$ denoting the load factor. The control objective is to track the reference load factor $n_{z,\cmd}$. 
The reference value is provided by an outer loop flight path and airspeed controller. The command can be reasonably assumed as an arbitrary norm-bounded signal, i.e., $n_{z,\cmd}\in L_2$; a common assumption to assess tracking performance. It is filtered through a low-pass filter $G_\mathrm{F}$ with a roll-off at $7$ rad/s to account for the control bandwidth of the outer loop controller. %For more information on the control architecture and synthesis procedure, the interested reader is referred to \cite{Bertran2025}. 
The load factor controller is
\begin{equation}
	K_{n_z} = \frac{-726.7s^2 - 1.308\cdot10^4s - 5.831\cdot10^4}{s^3 + 57.4s^2 + 810s}. 
\end{equation}
The elevator actuator is modeled by a first order lag $G_{\act}$ with a time constant $T_{\act} = 0.02$ s. Uncertainties in the system dynamics are modeled as a multiplicative dynamic uncertainty $\Delta$ at the actuator input. The uncertainty $\Delta$ is a SISO LTI system with $\norm{\Delta}_\infty\le b_\Delta$. For the presented example the uncertainty level is set to $60\%$, i.e., $b_\Delta = 0.6$. The uncertain closed loop interconnection of the analysis problem is shown in Fig.~{\ref{fig:BlkExt}}, with the performance output  $e = n_{z,\cmd} - n_z$.

%is designed for an approximation of the aircraft's short period using standard loop-shaping methods and consists of an proportional gain, integral boost, high frequency roll-off and a load component. It's transfer function is 

%The control error $e = n_{z,cmd} - n_z$ is the performance output of the analysis problem, where $n_{z,cmd}$ is the reference. For this example, the command is considered as an arbitrary signal in $L_2$, and therefore filtered by a low-pass with a roll-off at $7$~rad/s to represent the control bandwidth of the outer-loop controller. The elevator actuator dynamic is modeled by a first order lag $G_{Act}$ with a time constant $T_{act} = 0.02s$ and corrupted by a multiplicative dynamic uncertainty $\Delta$ at the actuator input. The uncertainty $\Delta$ is a SISO LTI system with $\norm{\Delta}_\infty\le b_\Delta$. For the presented example the uncertainty level is set to $60\%$, i.e., $b_\Delta = 0.6$. The uncertain closed-loop interconnection of the analysis problem is shown in Fig.~{\ref{fig:BlkExt}}.

\begin{figure}[h!]
	\centering \tikzstyle{blockdiag}	= [node distance=5mm, >=stealth', semithick]
\tikzstyle{block}		= [draw, rectangle, minimum width=0.8cm, minimum height=.7cm, align=center]

\begin{tikzpicture}[blockdiag]
	
% Main Block Order (From AC to the left)  
\node[block](AC){$G_{\UAV}$}; 									
\node[block, left = of AC, xshift = -2pt](Act){$G_{\act}$}; 
\node[sum, left = of Act, xshift = 5pt](SumUnc){};
\node[connector, left = of SumUnc](RefUnc){};
\node[connector2, left = of RefUnc, xshift = -3pt](ConUnc){}; ; 
\node[block, left = of ConUnc, xshift = 5pt](Ctrl){$K_{n_z}$};
\node[connector2, left = of Ctrl, xshift = 5pt](OutE){};
\node[sum, left = of OutE, xshift = 7pt](SumE){}; 
\node[block, left = of SumE, xshift =15pt, yshift = 17pt](LowPass){$G_\mathrm{F}$}; 

% Actuator to Aircraft
\draw[->]($(Act.north east)!0.66!(Act.south east)$) -- ($(AC.north west)!0.66!(AC.south west)$)node[below, pos = 0.5]{$\delta_\text{e}$}; 

% delta_e flow 
\draw[->](Ctrl.east) --(SumUnc.west);
\draw[->](SumUnc.east) -- (Act.west); 

% theta flow
\draw[->](AC.east) -| +(0.8cm, -0.7cm)node[left, pos=0.75]{$n_z$} -| (SumE.south)node[right, pos =0.95]{-};
\draw[->](SumE.east) -- (Ctrl.west); 
\draw[->](OutE)-- +(0cm,0.8cm)node[above, pos=1]{$e$};
\draw[->](LowPass.south) |- (SumE); 
\draw[<-](LowPass.north) -- +(0,0.35cm)node[above, pos = 1]{$n_{z,\cmd}$}; 

% uncertainty block 
\node[block, above = of RefUnc, yshift = -5pt](Delta){$\Delta$}; 
\draw[->](ConUnc.north) |- (Delta.west);
\draw[->](Delta.east) -| (SumUnc.north); 

% wind uncertainty block 
\node[block, above = of AC, yshift = -5pt](Wind){$\Delta_\wind$};
\draw[->] ($(AC.north east)!0.25!(AC.south east)$) -- +(0.3cm,0) |- (Wind.east)node[right, pos=0.2]{$x_{\delta_\wind}$};
\draw[<-] ($(AC.north west)!0.33!(AC.south west)$) -- +(-0.3cm,0) |- (Wind.west);  
%\draw[<-]($(AC.north west)!0.33!(AC.south west)$) -| +(-0.3cm,0.5cm)node[above, pos=1]{$\delta_\text{w}$}; 

\end{tikzpicture}  
	\caption{Analysis setup for load factor tracker.}
	\label{fig:BlkExt}
\end{figure}
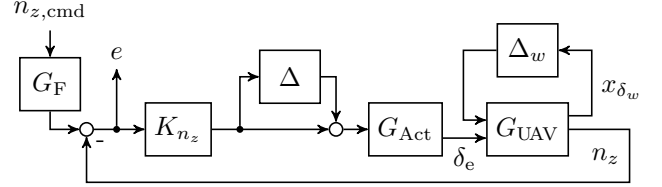

%The considered wind disturbances resembles an average day in Dresden, Germany based on meteorologic measurements. The maximum gust amplitude of the turbulent flow is {\color{red}xxx}, thus the wind can be described by an artifically fast time-varying disturbance signals as provided in Example~\ref{exmp:tv}, i.e. $\delta_\text{w} \in IQC(\Psi_{\delta}, M_\delta)$, with the factorization $\Psi_\delta = \bsmtx 1 & 0 \\ 0 & 1 \esmtx$ and parameterization $M_\delta= \bsmtx b_{\delta_\text{w}} X & 0\\ 0 & -X \esmtx$. 

\subsection{Robustness Analysis}

The analysis horizon is set to $T=3$ s. This is motivated by the fact that load factor variations mostly occur during turn maneuvers. In an urban scenario, these are usually completed within a few seconds. 
%The finite horizon induced $L_2$ norm from input to performance output is considered. 
To apply the analysis condition in Theorem \ref{thm:UncInit} and calculate the bound on the worst case $\norm{e}_{2[0,T]}$, the interconnection in Fig.~\ref{fig:BlkExt} must be transferred into the IQC framework as detailed in Section~\ref{sec:Back}. The model uncertainty $\Delta$ is represented by the an IQC resembling a SISO LTI dynamic uncertainty, see, e.g., \cite{Veenman2016}. This means, $\Delta \in IQC(\Psi, M)$ with $M= \bsmtx b_\Delta X & 0\\ 0 & -X \esmtx$ and $\Psi= \bsmtx \psi_\nu & 0 \\ 0 & \psi_\nu \esmtx$. The basis function is chosen as $\psi_\nu =\bsmtx 1 & \frac{1}{s+1} \esmtx^T$.\\
%The commanded signal $n_{z,cmd}$ is well represented by an arbitrary norm bounded worst-case signal analog to the signal $d$ in Sections~\ref{sec:Back} and \ref{sec:SigIQC}. 
The turbulent wind $\delta_w$ is approximated by an arbitrarily fast time-varying disturbance signal which can be modeled with the IQC given in Example~\ref{exmp:tv}, i.e. $\Delta_\wind \in IQC(\Psi_{\wind}, M_\wind)$, with the factorization $\Psi_\wind = \bsmtx 1 & 0 \\ 0 & 1 \esmtx$ and parameterization $M_\wind= \bsmtx b_{\delta_\wind} X & 0\\ 0 & -X \esmtx$.
A third-order polynomial basis functions represents $P$. The conditions of Theorem~\ref{thm:UncInit} are evaluated for a finite number of LMIs along an equidistant time grid with a step size of $0.1\,\text{s}$. The resulting semi-definite program is solved using Matlab's \texttt{lmilab}. A bisection over $\gamma$ yields the performance gain $\gamma = 4.98$. 

\begin{figure}[ht!]
	\centering % This file was created by matlab2tikz.
%
%The latest updates can be retrieved from
%  http://www.mathworks.com/matlabcentral/fileexchange/22022-matlab2tikz-matlab2tikz
%where you can also make suggestions and rate matlab2tikz.
%
\definecolor{mycolor1}{rgb}{0.00000,0.44700,0.74100}%
\begin{tikzpicture}

\begin{axis}[%
width=0.75\columnwidth,
height=3cm,
at={(0,0)},
scale only axis,
xmin=0,
xmax=0.8,
ymin=0,
ymax=40,
ytick = {0,10,20, 30, 40},
axis background/.style={fill=white},
xlabel = {Actuator Uncertainty Level $b_\Delta$},
ylabel style={yshift = -10pt},
ylabel = {Performance Value $\gamma$},
xmajorgrids,
ymajorgrids, 
]
\addplot [color=mycolor1, forget plot, line width=2]
  table[row sep=crcr]{%
		0.001	1.16779932879094\\
		0.05	1.24631378179468\\
		0.1	1.33721364900412\\
		0.15	1.44311721223326\\
		0.2	1.56783179910707\\
		0.25	1.71664050907598\\
		0.3	1.89551285585228\\
		0.35	2.11544740641048\\
		0.4	2.39208003584021\\
		0.45	2.75089198259889\\
		0.5	3.23741899300301\\
		0.55	3.92330274963042\\
		0.6	4.97567947635602\\
		0.61	5.25895010178035\\
		0.62	5.57753146096559\\
		0.63	5.93859219151951\\
		0.64	6.34990953428662\\
		0.65	6.82756190607423\\
		0.66	7.38580881779257\\
		0.67	8.05057187636304\\
		0.68	8.86038772959242\\
		0.69	9.86666484730352\\
		0.7	11.1721482051917\\
		0.71	12.9496753928469\\
		0.72	15.5221222863622\\
		0.73	19.6558525641567\\
		0.74	27.4662435031542\\
		0.75	47.7777651687658\\
		0.76	100\\
};
\end{axis}
\end{tikzpicture}%
	\caption{Robust performance for different actuator uncertainty levels.}
	\label{Fig:GammaOverUnc}
\end{figure}
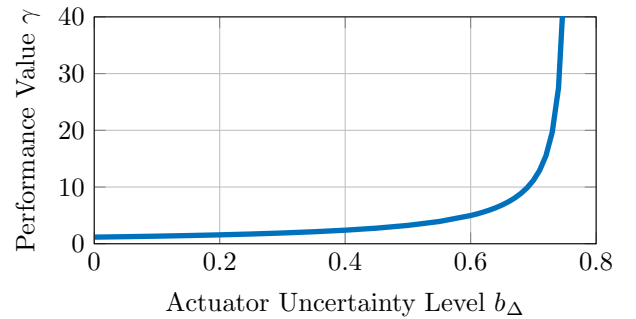	

The performance analysis was repeated for a total of $30$ different uncertainty levels distributed between $0$ and $80\%$, i.e., $b_\Delta \in [0\; 0.8]$. The wind disturbance remains bounded by the previously defined signal IQC $\Delta_\wind$ with $b_{\delta_\wind} = 5.5$ m/s. The results are depicted in Fig.~\ref{Fig:GammaOverUnc} and show the expected exponential increase in the performance value~$\gamma$, thus verifying that traditional uncertainty descriptions are captured correctly within the novel approach.  

\section{Conclusion}

This paper contributes a novel worst-case performance analysis that incorporates the specific characteristics of partially known disturbances while simultaneously considering worst-case norm-bounded inputs. The partially known disturbances are accounted for through variations of an internal signal of the system which is bounded by an appropriate signal integral quadratic constraints. This approach removes restrictions imposed by frameworks relying on induced norms for the analysis of systems under mixed disturbances, such as one-to-one correspondence between specific performance inputs and outputs. The corresponding analysis condition provides an upper bound on the $L_2[0,T]$ norm of the performance output, which establishes a relationshipt to classical induced norm approaches. The approach is applied to the path-tracking task of an UAV. %The comparison to a standard worst-case induced norm IQC analysis verifies the novel approach.

%\begin{ack}
%This research was funded by the German Federal Ministry for Economic Affairs and Climate Action under grant number 20Y2109E and partially supported by the European Union under Grant No. 101153910. 
%The responsibility for the content of this paper is with its authors. The views and opinions expressed do not necessarily reflect those of the respective authorities. The financial support is greatly appreciated. 
%\end{ack}
\bibliography{References_FiniteHorizonMixedDisturbanceSignal_IQC}             % bib file to
%\bibliography{rocond25_artificalState_literature}             % bib file to produce the bibliography
% with bibtex (preferred)

%\begin{thebibliography}{xx}  % you can also add the bibliography by hand

%\bibitem[Able(1956)]{Abl:56}
%B.C. Able.
%\newblock Nucleic acid content of microscope.
%\newblock \emph{Nature}, 135:\penalty0 7--9, 1956.

%\bibitem[Able et~al.(1954)Able, Tagg, and Rush]{AbTaRu:54}
%B.C. Able, R.A. Tagg, and M.~Rush.
%\newblock Enzyme-catalyzed cellular transanimations.
%\newblock In A.F. Round, editor, \emph{Advances in Enzymology}, volume~2, pages
%  125--247. Academic Press, New York, 3rd edition, 1954.

%\bibitem[Keohane(1958)]{Keo:58}
%R.~Keohane.
%\newblock \emph{Power and Interdependence: World Politics in Transitions}.
%\newblock Little, Brown \& Co., Boston, 1958.

%\bibitem[Powers(1985)]{Pow:85}
%T.~Powers.
%\newblock Is there a way out?
%\newblock \emph{Harpers}, pages 35--47, June 1985.

%\bibitem[Soukhanov(1992)]{Heritage:92}
%A.~H. Soukhanov, editor.
%\newblock \emph{{The American Heritage. Dictionary of the American Language}}.
%\newblock Houghton Mifflin Company, 1992.

%\end{thebibliography}

%\appendix
%\section{A summary of Latin grammar}    % Each appendix must have a short title.
%\section{Some Latin vocabulary}              % Sections and subsections are supported  
% in the appendices.
\end{document}